\documentclass[sigconf,nonacm]{acmart}
\usepackage{booktabs,multirow,array}
\usepackage{amsmath}
\usepackage{xcolor,graphicx}
\usepackage{colortbl}
\usepackage{enumitem}
\usepackage{float}
\usepackage{tikz}
\usepackage{pifont}
\usetikzlibrary{arrows.meta,positioning,fit,calc,backgrounds}

\newtheorem{appcproposition}{Proposition}[section]

\setcopyright{none}

\begin{document}
\title{LESSRec: Level-Specific History Modeling for Semantic ID-Based Recommendation}
\author{Shangxin Yang}
\affiliation{
  \institution{Chongqing University}
  \city{Chongqing}
  \country{China}
}
\email{yangshangxin@stu.cqu.edu.cn}

\author{Zongwei Wang}
\affiliation{
  \institution{Chongqing University}
  \city{Chongqing}
  \country{China}
}
\email{zongwei@cqu.edu.cn}

\author{Jing Long}
\affiliation{
  \institution{Chongqing University}
  \city{Chongqing}
  \country{China}
}
\email{jing.long@cqu.edu.cn}

\author{Junliang Yu}
\affiliation{
  \institution{Griffith University}
  \city{Brisbane}
  \country{Australia}
}
\email{junl.yu@outlook.com}

\author{Min Gao}
\authornote{Corresponding author.}
\affiliation{
  \institution{Chongqing University}
  \city{Chongqing}
  \country{China}
}
\email{gaomin@cqu.edu.cn}
\renewcommand{\shortauthors}{Yang et al.}

\begin{abstract}
Semantic ID-based recommendation represents each item as a sequence of semantic tokens obtained through residual quantization. This paradigm generates the next item's tokens autoregressively from coarse to fine, entailing each prediction attending to historical tokens at all semantic levels. However, our inference-time analysis reveals that  historical tokens at the target-token's semantic level tend to be more informative, whereas other levels can introduce interference. Motivated by these findings, we propose LESSRec that separately encodes the historical token sequence at each semantic level and restricts attention over historical tokens to the sequence at the target token's semantic level during prediction. To complement these level-specific histories, LESSRec explicitly provides an item-level semantic sequence that represents each historical item's overall semantics in chronological order. The decoder accesses the level-aligned token sequence and the item-level semantic sequence through two separate cross-attention branches and fuses their contexts. We further introduce LESSRec+, which uses separate encoders for different semantic-level histories with the same total parameter count to reduce encoder computation. Extensive experiments across three public datasets and three tokenizers demonstrate the recommendation effectiveness of LESSRec and LESSRec+, with efficiency analyses highlighting the computational benefits of LESSRec+.
\end{abstract}


\maketitle

\section{Introduction}
 Recommender systems support personalized item discovery on Web platforms such as online retail services~\cite{smith2017two}. Generative recommendation has recently gained attention as a new paradigm that formulates item prediction as a sequence generation task ~\cite{guo2026onesug,zhou2025onerec,fu2026forge,yang2025sparsea}. Under this paradigm, semantic ID-based approaches have become particularly prominent, which predict items by generating their semantic IDs, each comprising an ordered sequence of discrete semantic tokens obtained through residual quantization~\cite{rajput2023recommender,wang2024learnable,zhu2024costa}. Within each item's semantic IDs, earlier tokens typically represent coarse-grained item semantics, while later tokens encode finer-grained information~\cite{wei2026cofireca,liu2025generative,zheng2024adaptingb}. This structure enables semantically similar items to share coarse-level prefix tokens, capturing their common semantics while retaining finer-grained distinctions.

Conventional semantic ID-based recommendation models adopt an encoder–decoder architecture~\cite{rajput2023recommender,liu2025generative,zhou2025onerec}. The encoder contextualizes historical semantic tokens through self-attention, allowing tokens from different items and semantic levels to exchange information. The decoder can then attend to encoded tokens from all historical items and semantic levels, regardless of the level being predicted. This design brings together two distinct sequential structures: tokens at the same level across interactions capture behavioral patterns at a consistent semantic granularity, whereas tokens within an item describe its progressive semantic refinement. Given these different roles, it is unclear whether information from all levels is equally useful for next token prediction. This motivates a key question: when predicting a token at a particular semantic level, which levels in the interaction history provide useful evidence?

\begin{figure}[t]
\centering
\begingroup
\edef\introPanelWidthA{\the\dimexpr\columnwidth*60/100\relax}
\edef\introPanelWidthBC{\the\dimexpr\columnwidth*39/100\relax}
\setbox4=\hbox{\parbox{\introPanelWidthA}{\centering\fontsize{7}{8}\selectfont\strut(a) History-access patterns\\for a level-2 target\strut}}
\setbox6=\hbox{\parbox{\introPanelWidthBC}{\centering\fontsize{7}{8}\selectfont\strut(b) Standard TIGER (shared encoding)\strut}}
\setbox8=\hbox{\parbox{\introPanelWidthBC}{\centering\fontsize{7}{8}\selectfont\strut(c) TIGER with level-isolated encoding\strut}}
\setbox0=\hbox{\includegraphics[width=\introPanelWidthA,trim=1.5bp 0bp 0bp 0bp,clip]{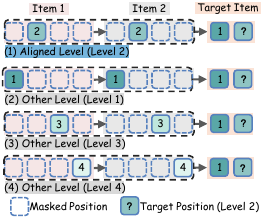}}
\edef\introSmallPanelHeight{\the\dimexpr(\ht0+\dp0+\ht4+\dp4-\ht6-\dp6-\ht8-\dp8-4pt)/2\relax}
\setbox0=\vbox{\offinterlineskip
\box0
\kern1pt
\box4}
\setbox2=\vbox{\offinterlineskip
\hbox to \introPanelWidthBC{\hfil\includegraphics[height=\introSmallPanelHeight]{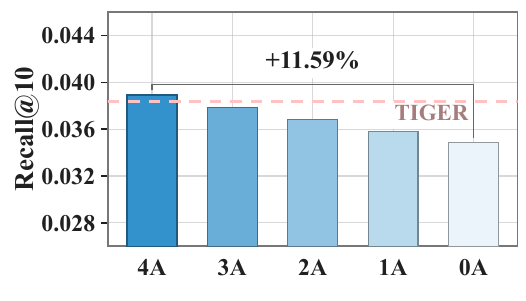}\hfil}
\kern1pt
\box6
\kern3pt
\hbox to \introPanelWidthBC{\hfil\includegraphics[height=\introSmallPanelHeight]{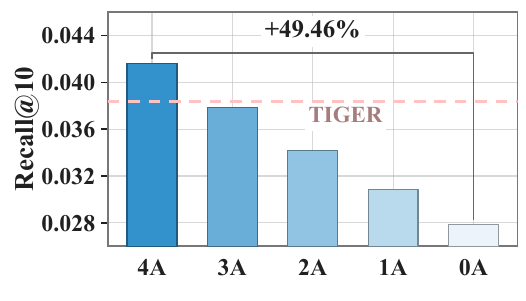}\hfil}
\kern1pt
\box8}
\typeout{Figure 1 panel heights: A=\the\dimexpr\ht0+\dp0\relax; BC=\the\dimexpr\ht2+\dp2\relax}
\noindent\resizebox{\columnwidth}{!}{\box0\hspace{3pt}\box2}
\endgroup
\caption{Same-level vs. cross-level history access for prediction in TIGER under standard and level-isolated encoding (Recall@10 on Games). Under inference-time cross-attention masking, $(m)$A denotes $(m)$ of four target positions using same-level access (e.g., (1) in (a)); the remaining positions use level-mismatched access (e.g., (2), (3), and (4) in (a)). Bars average configurations with the same $(m)$. Dashed lines mark the original TIGER baseline; percentages show relative gains of 4A over 0A.}
\Description{A schematic on the left and two vertically stacked bar charts on the right. Panel a illustrates level-aligned and level-mismatched history access when predicting a level-2 token. Panels b and c show Recall at 10 on Games for standard TIGER and TIGER trained with level-isolated encoder self-attention, respectively. Each chart compares four, three, two, one, and zero aligned target positions. Both dashed lines show the same standard TIGER baseline with unrestricted history access during training and inference.}
\label{fig:alignment_ablation}
\end{figure}

To understand how the decoder uses historical information from different semantic levels for prediction, we conduct an inference-time masking experiment on TIGER \cite{rajput2023recommender}, a representative encoder--decoder generative recommender. We keep all parameters fixed and modify only the decoder's cross-attention mask, restricting each target position to tokens at a single semantic level in the interaction history (Figure~\ref{fig:alignment_ablation}(a)). Across masking configurations, the number of visible tokens remains constant, but performance improves as more target positions attend to their matching semantic levels (Figure~\ref{fig:alignment_ablation}(b)). With all positions aligned, each prediction accesses only one quarter of the history tokens, yet the recommendation performance slightly outperforms that of original TIGER, which uses unrestricted history access during both training and inference. These results suggest that explicitly matching historical access to the target semantic level can be more effective than providing access to all levels.

However, TIGER's encoder mixes information across semantic levels before decoding, so attending to positions at one level does not isolate information from that level. To examine whether this mixing affects the observed alignment benefit, we train a TIGER variant that restricts encoder self-attention to tokens at the same semantic level, leaving the rest of the architecture unchanged. We then repeat the inference-time masking experiment on this variant. Figure~\ref{fig:alignment_ablation}(c) shows a larger contrast between aligned and mismatched access. Within each model, fully aligned access improves Recall@10 over fully mismatched access by 11.59\% under standard encoding and 49.46\% under level-isolated encoding. With cross-level mixing restricted during encoding, the larger gap between aligned and mismatched access provides clearer evidence that same-level history is more useful for prediction than cross-level history.

Motivated by these findings, we propose LESSRec (LEvel-Specific history modeling for Semantic ID-based Recommendation) to better use level-specific historical evidence. Specifically, Level-Isolated Self-Attention separately contextualizes historical tokens at each semantic level, preserving same-level associations across items. During prediction, the decoder uses Level-Aligned Cross-Attention to attend to historical semantic tokens at the same level as the target token. However, a single level offers only a partial semantic view, and relying on it alone may overlook how different levels complement one another in describing an item. To preserve this complementarity, we combine the contextualized states from different levels within each historical item to construct an Item-Level Semantic Sequence. The decoder accesses this sequence through a separate cross-attention branch and fuses its context with the level-aligned context, combining level-specific historical evidence with complementary item semantics for prediction.

LESSRec implements this organization through attention masks while sharing encoder layers across semantic levels. To avoid unnecessary computation over masked cross-level interactions and give each level its own encoder-layer parameters, we further propose LESSRec+. LESSRec+ partitions the existing encoder layers into level-specific groups, each processing only the historical semantic tokens at its assigned level. This reduces encoder computation and intermediate activation storage while retaining the same total parameter count as LESSRec.

Our main contributions are as follows:
\begin{itemize}[leftmargin=*,nosep]
\item To the best of our knowledge, we are the first to identify the predictive advantage of same-level history in semantic ID-based generative recommendation.
\item We propose LESSRec to jointly organize same-level historical associations across items and complementary semantic information within each item for prediction at each semantic level. LESSRec+ further reorganizes the encoder into level-specific groups, reducing computation while retaining the same parameter count as LESSRec.
\item Experiments on three datasets with three tokenizers demonstrate the recommendation effectiveness of LESSRec and LESSRec+. Component ablations support the proposed design, while efficiency analyses show the computational benefits of LESSRec+.
\end{itemize}

\section{Preliminaries}
\label{sec:preliminary}
\label{sec:problem}

Let $\mathcal{I}=(v_1,\ldots,v_T)$ be a user's interaction history in chronological order. Each item $v_i$ is represented by its semantic IDs, obtained through the residual quantization process described in Appendix~\ref{app:tokenization}:
\[
\mathbf{s}_i=(s_{i,1},\ldots,s_{i,L}),
\]
where $s_{i,r}$ is the semantic token at level $r\in\{1,\ldots,L\}$. The item-to-semantic-ID mapping remains fixed after tokenization. We concatenate the semantic IDs of historical items in interaction order to form the input:
\begin{equation}
\mathbf{x}=(s_{1,1},\ldots,s_{1,L},\ldots,s_{T,1},\ldots,s_{T,L},\mathrm{EOS}).
\label{eq:input_sequence}
\end{equation}
The decoder autoregressively generates the next item's semantic IDs, $Y_1,\ldots,Y_L$, followed by EOS. At target position $k$, the prediction is conditioned on the history $\mathbf{x}$ and the preceding target tokens $C_k=Y_{<k}$.

For each semantic level $r$, we define the same-level historical sequence as
\[
H_r=(s_{1,r},s_{2,r},\ldots,s_{T,r}),
\]
which contains the level-$r$ tokens of historical items in interaction order. To specify which tokens are visible under semantic-level masking, let $p(i,r)$ denote the zero-based position of $s_{i,r}$ in the input sequence $\mathbf{x}$:
\[
p(i,r)=L(i-1)+(r-1).
\]
The set $\mathcal{P}_r=\{p(i,r):1\leq i\leq T\}$ collects the input positions of tokens in $H_r$. Including the end token EOS at $p_{\mathrm{end}}=LT$, we define $\mathcal{P}_{\mathrm{valid}}=\left(\bigcup_{r=1}^{L}\mathcal{P}_r\right)\cup\{p_{\mathrm{end}}\}$ as the set of all valid input positions, excluding padding.

\section{Method}
\label{sec:method}
To address cross-level interference, we propose LESSRec (Figure~\ref{fig:lessrec_framework}). It comprises level-isolated self-attention (Section~\ref{sec:slf}), level-aligned cross-attention (Section~\ref{sec:lmca}), and an item-level semantic sequence (Section~\ref{sec:icm}). After describing these components and their integration, we introduce LESSRec+ with level-specific encoders (Section~\ref{sec:slirec_plus}).

\begin{figure}[t]
\centering
\includegraphics[width=\columnwidth]{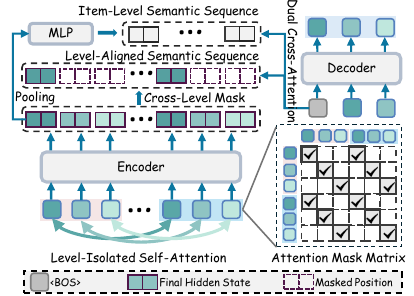}
\caption{Overview of LESSRec. For visual clarity, each item is represented by three semantic tokens, and \texttt{<EOS>} tokens are omitted. Shades of the same color distinguish semantic levels, from coarse (dark) to fine (light). In the attention mask matrix, \ding{51} indicates visible positions and blank cells indicate masked positions.}
\Description{A single-column encoder--decoder framework showing level-isolated self-attention and its mask matrix, cross-level masking of the encoded history, within-item pooling and an MLP to form an item-level semantic sequence, and dual cross-attention to both sequences.}
\label{fig:lessrec_framework}
\end{figure}

\subsection{Level-Isolated Self-Attention (LIS)}
\label{sec:slf}
To help the model learn same-level contextual information without interference from other levels, we introduce Level-Isolated Self-Attention (LIS). LIS restricts each historical semantic token to attending only to tokens at the same level across items:
\begin{equation}
M^{\mathrm{enc}}_{qj}=
\begin{cases}
0,&q\in\mathcal{P}_r,\ j\in\mathcal{P}_r,\\
0,&q=p_{\mathrm{end}},\ j\in\mathcal{P}_{\mathrm{valid}},\\
-\infty,&\text{otherwise}.
\end{cases}
\label{eq:encoder_mask}
\end{equation}
Here, $q$ and $j$ are the zero-based positions of the query token and the token being attended to in $\mathbf{x}$, respectively. This mask is applied in every encoder layer so that historical semantic tokens exchange information only within their corresponding level sequences. Because EOS serves as a shared terminator for the complete input sequence, we allow its encoder state to attend to all valid positions, excluding padding. However, historical semantic tokens cannot attend to EOS: $M^{\mathrm{enc}}_{q,p_{\mathrm{end}}}=-\infty$ for all $q\in\bigcup_{r=1}^{L}\mathcal{P}_r$. Thus, information aggregated by EOS cannot flow back into these token states, and the historical sequences remain semantically isolated across levels during encoding.

Let $\mathbf{h}_{i,r}\in\mathbb{R}^{d}$ denote the final encoded state of the level-$r$ token of item $i$. Each state incorporates historical context from $H_r$ without mixing information from tokens at other levels.

\subsection{Level-Aligned Cross-Attention (LAC)}
\label{sec:lmca}
Although LIS separates historical semantic levels during encoding, unrestricted cross-attention still allows the decoder to combine information from different levels. To focus each prediction on same-level history, we introduce Level-Aligned Cross-Attention (LAC): when predicting a level-$r$ target token, the level-aligned branch attends only to historical level-$r$ tokens.

The cross-attention mask for the level-aligned branch is
\begin{equation}
M^{\mathrm{cross}}_{tj}=
\begin{cases}
0,&0\leq t<L,\ j\in\mathcal{P}_{t+1}\cup\{p_{\mathrm{end}}\},\\
0,&t=L,\ j\in\mathcal{P}_{\mathrm{valid}},\\
-\infty,&\text{otherwise}.
\end{cases}
\label{eq:raw_cross_mask}
\end{equation}
Here, $t\in\{0,\ldots,L\}$ denotes the zero-based decoder query position, and $j$ denotes the zero-based position of an encoded history token. Positions $t=0,\ldots,L-1$ predict the $L$ semantic tokens, and $t=L$ predicts the end token.

LIS and LAC act at complementary stages: LIS prevents cross-level mixing in historical semantic-token states, and LAC selects the matching level's states during prediction. Together, they enable the decoder to use contextual information from same-level history without attending to historical semantic tokens at other levels. Appendix~\ref{app:level_theory}  provides a theoretical analysis of same-level history access under simplified assumptions.

\subsection{Item-Level Semantic Sequence (ILS)}
\label{sec:icm}
LIS and LAC help the decoder use same-level history, but attending to a single semantic level limits access to the combined semantics of each historical item. We therefore construct an Item-Level Semantic Sequence (ILS), in which each representation combines an item's encoded semantic tokens and the representations follow the user's item interaction order. For each historical item $v_i$, we construct an item-level representation by averaging the LIS-encoded representations of its $L$ semantic tokens:
\begin{equation}
\mathbf{z}_i=\frac{1}{L}\sum_{r=1}^{L}\mathbf{h}_{i,r}.
\label{eq:item_pool}
\end{equation}
Since each $\mathbf{h}_{i,r}$ already incorporates same-level historical context through LIS, the pooled representation $\mathbf{z}_i$ combines the item's semantics with this context. We therefore construct the item-level sequence from these contextualized representations without an additional self-attention step.

To refine each item representation while retaining the pooled information, we apply a residual projection:
\begin{gather}
\mathbf{m}_i=\mathbf{z}_i+\operatorname{MLP}\!\left(\operatorname{LN}(\mathbf{z}_i)\right),
\label{eq:item_mlp}\\
\operatorname{MLP}\!\left(\operatorname{LN}(\mathbf{z}_i)\right)
=W_2\operatorname{GELU}\!\left(W_1\operatorname{LN}(\mathbf{z}_i)+\mathbf{b}_1\right)+\mathbf{b}_2.
\label{eq:item_mlp_expansion}
\end{gather}
Here, LN denotes layer normalization, and MLP is a two-layer feed-forward network with GELU (Gaussian error linear unit) activation. The residual connection adds the learned nonlinear update to the pooled representation $\mathbf{z}_i$. The resulting sequence $\mathbf{M}=(\mathbf{m}_1,\ldots,\mathbf{m}_T)$ represents historical items in interaction order and provides item-level context alongside the level-aligned history.

\subsection{Dual Cross-Attention Fusion}
\label{sec:memory_fusion}
\label{sec:item_position}
To combine information from the level-aligned semantic sequence and the item-level semantic sequence constructed above, we design dual cross-attention fusion. Let $\mathbf{R}$ denote the encoded historical token sequence produced by LIS and $\mathbf{M}$ the item-level semantic sequence. For the same decoder query $\mathbf{Q}$, two multi-head cross-attention branches~\cite{vaswani2017attention} separately attend to these sequences:
\begin{align}
\mathbf{o}_{\mathrm{level}}
&=\operatorname{MHA}(\mathbf{Q},\mathbf{R},\mathbf{R};
 B_{\mathrm{level}}+M^{\mathrm{cross}}),\label{eq:raw_branch}\\
\mathbf{o}_{\mathrm{item}}
&=\operatorname{MHA}(\mathbf{Q},\mathbf{M},\mathbf{M};
 B_{\mathrm{item}}+M^{\mathrm{item}}).\label{eq:item_branch}
\end{align}
Here, $M^{\mathrm{cross}}$ is the level-aligned cross-attention mask defined in Equation~\ref{eq:raw_cross_mask}, and $M^{\mathrm{item}}$ masks padded positions in the item-level sequence. $B_{\mathrm{level}}$ is the original relative position bias that distinguishes semantic-token positions in the level-specific history.

This token-based bias does not match the item-level semantic sequence: consecutive historical items occupy adjacent positions in $\mathbf{M}$, whereas their corresponding same-level tokens are $L$ positions apart in the input semantic-token sequence. We therefore introduce $B_{\mathrm{item}}$ to model positions in the item-level semantic sequence:
\begin{equation}
B_{\mathrm{item}}^{(h)}(t,j)
=E_{\mathrm{item}}^{(h)}
\left[\operatorname{bucket}(j-t)\right].
\label{eq:item_position_bias}
\end{equation}
Here, $t$ is the decoder query position, $j\in\{0,\ldots,T-1\}$ indexes historical items in interaction order, and $h$ indexes attention heads. The function $\operatorname{bucket}(\cdot)$ groups relative offsets into distance ranges and returns the corresponding index, which is used to retrieve a bias from the learned table $E_{\mathrm{item}}$~\cite{raffel2020exploringb}. Defining these offsets over item positions provides the decoder with positional information aligned with the chronological order of historical item interactions.

We fuse the contexts from the level-aligned history and the item-level semantic sequence through a weighted sum of the two branch outputs:
\begin{equation}
\mathbf{o}=(1-\alpha)\mathbf{o}_{\mathrm{level}}
+\alpha\mathbf{o}_{\mathrm{item}}.
\label{eq:dual_mix}
\end{equation}
Here, $\alpha\in[0,1]$ is a hyperparameter controlling the relative contributions of the level-aligned history and the item-level semantic sequence. This fusion retains information from the historical sequence at the target semantic level while incorporating the overall semantics of items in interaction order.

\subsection{Training and Inference}
\label{sec:training}
We jointly optimize the encoder, decoder, residual item projection, and item-level positional bias using the average autoregressive negative log-likelihood:
\begin{equation}
\mathcal{L}
=-\frac{1}{(L+1)|\mathcal{D}|}
\sum_{(\mathbf{x},\mathbf{y})\in\mathcal{D}}
\sum_{k=1}^{L+1}
\log p_{\theta}(y_k\mid y_{<k},\mathbf{x}),
\label{eq:objective}
\end{equation}
where $\mathcal{D}$ contains history--target pairs and $\mathbf{y}$ consists of the next item's $L$ semantic tokens followed by the end token. During training, each prediction is conditioned on the preceding tokens in the ground-truth target sequence. No auxiliary objective is required, and the tokenizer and item-to-semantic-ID mapping remain fixed.

At inference, the decoder autoregressively generates the next item's semantic IDs using trie-constrained beam search over valid identifiers~\cite{cao2020autoregressive}.

\subsection{LESSRec+: Level-Specific Encoders}
\label{sec:slirec_plus}
Level-Isolated Self-Attention provides a simple implementation of semantic-level isolation by masking encoder self-attention. However, all semantic levels still share the same encoder layers, and attention scores are computed for all pairs of historical semantic tokens before masking. LESSRec+ instead divides the original encoder layers into level-specific groups, each encoding only the historical token sequence at its assigned semantic level (Figure~\ref{fig:slirec_plus_encoder}).

\begin{figure}[t]
\centering
\includegraphics[width=\columnwidth]{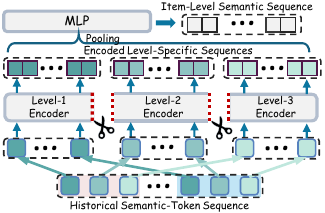}
\caption{Encoder-side illustration of LESSRec+. Its decoder uses the same dual cross-attention as in Figure~\ref{fig:lessrec_framework}.}
\Description{The encoder side of LESSRec+, showing historical semantic tokens grouped into level-specific sequences, separate encoders, and pooling followed by an MLP to construct an item-level semantic sequence. The decoder is omitted and retains the dual cross-attention design of LESSRec.}
\label{fig:slirec_plus_encoder}
\end{figure}

For each semantic level $r$, we form the level-specific input $\mathbf{x}^{(r)}=(s_{1,r},\ldots,s_{T,r},\mathrm{EOS})$. The original encoder has $D$ layers. We divide them into $L$ disjoint groups, assigning $D_r$ consecutive layers to semantic level $r$ with $\sum_{r=1}^{L}D_r=D$. Each group forms an encoder $E_r$ that processes only $\mathbf{x}^{(r)}$. By reusing the existing embeddings and partitioning the encoder layers, each encoder independently models its level's history without increasing the total parameter count.

After encoding, we restore each historical token state to its position $p(i,r)$ in the original layout. The level-specific end-token states are averaged to obtain a single encoder EOS state:
\begin{equation}
\widetilde{\mathbf{h}}_{\mathrm{end}}
=\frac{1}{L}\sum_{r=1}^{L}\mathbf{h}^{(r)}_{\mathrm{end}}.
\label{eq:plus_end_pool}
\end{equation}
The restored states are used to construct Item-Level Semantic Sequence and supply the level-aligned token history. All other modules and the training objective remain as defined above.

Each encoder now processes $T+1$ positions. Summed over all encoder layers, the attention-score matrix size is $D(T+1)^2$, compared with $D(LT+1)^2$ for masked attention in LESSRec. Moreover, the level-specific encoders have no computational dependencies on one another and can therefore run in parallel, reducing the longest sequential path through encoder layers from $D$ to $\max_r D_r$. Processing shorter input sequences reduces encoder computation and intermediate activation storage without adding parameters; we compare the models in terms of computational complexity, GPU memory usage, and efficiency, and analyze encoder layer allocation in the experimental section.

\section{Experiments}
\label{sec:experiments}
We organize the experiments around five research questions:
\begin{itemize}[leftmargin=*,nosep]
\item \textbf{RQ1:} How do LESSRec and LESSRec+ perform against recommendation baselines across datasets and tokenization methods?
\item \textbf{RQ2:} How do the three main components (LIS, LAC, and ILS) affect recommendation performance?
\item \textbf{RQ3:} How do different settings ($\alpha$, the number of semantic tokens per item, and encoder layer allocation) affect recommendation performance?
\item \textbf{RQ4:} How do the models compare in computational complexity, parameter count, peak GPU memory usage, and runtime?
\item \textbf{RQ5:} How does interaction-history length affect the recommendation performance and encoding efficiency of LESSRec and LESSRec+?
\end{itemize}

\subsection{Experimental Settings}
\label{sec:setup}

\subsubsection{Datasets}
{\emergencystretch=1em
We evaluate LESSRec on Games (Video Games), Software, and Beauty. Games and Software use Amazon Reviews 2023~\cite{hou2026bridginga}; its source collection covers May 1996 to September 2023. Beauty comes from an earlier Amazon review collection~\cite{ni2019justifying}. Following previous studies~\cite{rajput2023recommender,wang2024learnable}, we iteratively remove users and items with fewer than five interactions until the remaining data satisfy the 5-core requirement. We restrict the number of items in a user’s history to 20. Table~\ref{tab:datasets} reports the statistics after preprocessing.}

\begin{table}[H]
\caption{Dataset statistics after final preprocessing. Avg. length is the mean number of interactions per user.}
\label{tab:datasets}\centering\small
\resizebox{\columnwidth}{!}{\begin{tabular}{lrrrrr}\toprule
Dataset & Users & Items & Interactions & Sparsity & Avg. length\\\midrule
Games & 26,443 & 10,512 & 202,250 & 99.9272\% & 7.65\\
Software & 22,225 & 4,654 & 200,915 & 99.8058\% & 9.04\\
Beauty & 22,363 & 12,101 & 198,502 & 99.9266\% & 8.88\\\bottomrule
\end{tabular}}
\end{table}

\subsubsection{Baselines}
We compare LESSRec with conventional sequential recommenders and semantic ID-based generative recommenders: \textbf{Caser}~\cite{tang2018personalized} models interaction sequences with convolutional filters. \textbf{GRU4Rec}~\cite{hidasi2016session} uses recurrent neural networks to model sequential interactions. \textbf{BERT4Rec}~\cite{sun2019bert4rec} learns bidirectional interaction representations through masked-item prediction. \textbf{SASRec}~\cite{kang2018self} models sequential interactions using self-attention. \textbf{RPG}~\cite{hou2025generating} generates the tokens of long, unordered semantic IDs in parallel. \textbf{CARE}~\cite{lin2026bringing} combines progressive history encoding with query-anchored reasoning to mitigate bias amplification. \textbf{DiscRec}~\cite{liu2025discreca} is a structure-aware recommender that combines intra-item attention with gated embedding fusion. \textbf{TIGER}~\cite{rajput2023recommender} uses RQ-VAE to quantize item representations into semantic IDs and autoregressively generates the next item's semantic IDs from the interaction history. We instantiate and compare TIGER, LESSRec, and LESSRec+ using three tokenizers. \textbf{RQ-VAE}~\cite{lee2022autoregressive} learns hierarchical codebooks by recursively quantizing residuals of item semantic representations. The \textbf{RQ-Kmeans} setting uses K-means-guided conditional residual quantization with trainable subcodebooks. Appendix~\ref{app:tokenization} details both implementations. \textbf{LETTER}~\cite{wang2024learnable} aligns quantized item representations with collaborative item embeddings and applies diversity regularization to alleviate code-assignment bias.

\subsubsection{Evaluation Settings}
We adopt a leave-one-out evaluation protocol. Recommendation quality is measured using Recall@5, Recall@10, NDCG@5, and NDCG@10. Within each tokenization setting, comparisons among TIGER, LESSRec, and LESSRec+ use the same data partitions, fixed semantic IDs, history length, candidate index, decoding settings, and checkpoint-selection rule.

\subsubsection{Generation of semantic IDs}
\label{sec:sid_generation}
For each dataset, semantic IDs are generated offline and remain fixed during recommendation training and evaluation. For RQ-VAE, RQ-Kmeans, and LETTER, we use four semantic tokens per item ($L=4$) and 256 codewords per codebook. Item semantic embeddings are extracted from text using Sentence-T5~\cite{ni2022sentence}. For LETTER, we use SASRec~\cite{kang2018self} to obtain collaborative item embeddings. For the RPG baseline, we use 32-token semantic IDs.

\subsubsection{Implementation Details}
TIGER, LESSRec, and LESSRec+ are implemented with Hugging Face Transformers using a randomly initialized T5-small backbone~\cite{raffel2020exploringb}, with six encoder and six decoder layers and a hidden dimension of 512. LESSRec+ allocates the six encoder layers across semantic levels with depths [2,2,1,1]. The item-memory MLP has an intermediate dimension of 256, and $\alpha=0.5$. We use AdamW with a learning rate of $5\times10^{-4}$, weight decay of 0.01, and a cosine schedule with a 1\% warmup ratio. The training batch size is 256 with two gradient accumulation steps. Each run uses a single NVIDIA GeForce RTX 4090 GPU (24 GB) for at most 200 epochs, stopping after 20 epochs without validation-loss improvement and selecting the checkpoint with the lowest validation loss. Inference uses a batch size of 32 and trie-constrained beam search with a beam size of 10. For the generative recommendation baselines, we utilize their official implementations. RPG uses a randomly initialized Transformer decoder with two layers, a hidden dimension of 448, and four attention heads. CARE uses the pretrained Qwen2.5-0.5B backbone~\cite{qwen2025qwen25}, retaining its 24 Transformer layers, hidden dimension of 896, and 14 attention heads. DiscRec is also implemented with a T5-small backbone~\cite{raffel2020exploringb}. For the conventional sequential recommendation baselines, we implement Caser, GRU4Rec, BERT4Rec, and SASRec using RecBole~\cite{zhao2021recbolea,zhao2022recbole}.

\subsection{Overall Performance (RQ1)}
\label{sec:results}

\begin{table*}[!t]
\caption{Main recommendation results on three public datasets. Score cells report mean performance over seeds $\{42,43,44\}$. Shaded ``Improv.'' rows report percentage gains over the baseline within the same tokenization group. Bold values indicate statistically significant improvements over the corresponding baseline under a paired t-test ($p < 0.05$).}
\label{tab:main}\centering\small
\setlength{\tabcolsep}{3.2pt}
\resizebox{\textwidth}{!}{\begin{tabular}{lrrrrrrrrrrrr}\toprule
&\multicolumn{4}{c}{Games}&\multicolumn{4}{c}{Software}&\multicolumn{4}{c}{Beauty}\\
\cmidrule(lr){2-5}\cmidrule(lr){6-9}\cmidrule(lr){10-13}
Method&R@5&R@10&N@5&N@10&R@5&R@10&N@5&N@10&R@5&R@10&N@5&N@10\\\midrule
Caser &0.0278&0.0456&0.0172&0.0229&0.1296&0.1920&0.0829&0.1029&0.0172&0.0299&0.0105&0.0146\\
GRU4Rec &0.0278&0.0441&0.0171&0.0224&0.1241&0.1815&0.0765&0.0951&0.0138&0.0234&0.0086&0.0117\\
BERT4Rec &0.0263&0.0415&0.0166&0.0215&0.1302&0.1898&0.0847&0.1040&0.0123&0.0226&0.0076&0.0110\\
SASRec &0.0226&0.0421&0.0116&0.0179&0.1093&0.1818&0.0606&0.0839&0.0178&0.0358&0.0095&0.0152\\
\midrule
RPG &0.0264&0.0383&0.0181&0.0219&0.0832&0.1080&0.0620&0.0700&0.0261&0.0390&0.0178&0.0219\\
CARE &0.0228&0.0461&0.0136&0.0211&0.1308&0.2266&0.0828&0.1135&0.0214&0.0465&0.0130&0.0211\\
DiscRec &0.0283&0.0395&0.0184&0.0219&0.1379&0.1959&0.0925&0.1113&0.0223&0.0365&0.0137&0.0183\\
\midrule
TIGER (RQ-VAE) &0.0260&0.0371&0.0172&0.0208&0.1276&0.1868&0.0842&0.1034&0.0201&0.0331&0.0125&0.0167\\
LESSRec (RQ-VAE) &0.0319&0.0457&0.0213&0.0257&0.1479&0.2092&0.1000&0.1198&0.0276&0.0440&0.0174&0.0227\\
\rowcolor{black!6}
\quad \textbf{\textit{Improv.}} & \textbf{+22.7\%} & \textbf{+23.2\%} & \textbf{+23.8\%} & \textbf{+23.6\%} & \textbf{+15.9\%} & \textbf{+12.0\%} & \textbf{+18.8\%} & \textbf{+15.9\%} & \textbf{+37.3\%} & \textbf{+32.9\%} & \textbf{+39.2\%} & \textbf{+35.9\%}\\
LESSRec+ (RQ-VAE) & 0.0335 & 0.0475 & 0.0225 & 0.0270 & 0.1559 & 0.2163 & 0.1067 & 0.1263 & 0.0297 & 0.0464 & 0.0187 & 0.0241\\
\rowcolor{black!6}
\quad \textbf{\textit{Improv.}} & \textbf{+28.8\%} & \textbf{+28.0\%} & \textbf{+30.8\%} & \textbf{+29.8\%} & \textbf{+22.2\%} & \textbf{+15.8\%} & \textbf{+26.7\%} & \textbf{+22.1\%} & \textbf{+47.8\%} & \textbf{+40.2\%} & \textbf{+49.6\%} & \textbf{+44.3\%}\\
\midrule
TIGER (RQ-Kmeans) &0.0316&0.0465&0.0212&0.0260&0.1297&0.1880&0.0838&0.1027&0.0228&0.0369&0.0145&0.0190\\
LESSRec (RQ-Kmeans) &0.0385&0.0568&0.0255&0.0314&0.1477&0.2076&0.0995&0.1189&0.0315&0.0489&0.0200&0.0256\\
\rowcolor{black!6}
\quad \textbf{\textit{Improv.}} & \textbf{+21.8\%} & \textbf{+22.2\%} & \textbf{+20.3\%} & \textbf{+20.8\%} & \textbf{+13.9\%} & \textbf{+10.4\%} & \textbf{+18.7\%} & \textbf{+15.8\%} & \textbf{+38.2\%} & \textbf{+32.5\%} & \textbf{+37.9\%} & \textbf{+34.7\%}\\
LESSRec+ (RQ-Kmeans) & 0.0387 & 0.0560 & 0.0259 & 0.0316 & 0.1555 & 0.2148 & 0.1061 & 0.1253 & 0.0343 & 0.0537 & 0.0218 & 0.0281\\
\rowcolor{black!6}
\quad \textbf{\textit{Improv.}} & \textbf{+22.5\%} & \textbf{+20.4\%} & \textbf{+22.2\%} & \textbf{+21.5\%} & \textbf{+19.9\%} & \textbf{+14.3\%} & \textbf{+26.6\%} & \textbf{+22.0\%} & \textbf{+50.4\%} & \textbf{+45.5\%} & \textbf{+50.3\%} & \textbf{+47.9\%}\\
\midrule
TIGER (LETTER) &0.0334&0.0513&0.0218&0.0276&0.1383&0.2027&0.0930&0.1138&0.0225&0.0375&0.0138&0.0187\\
LESSRec (LETTER) &0.0379&0.0584&0.0249&0.0315&0.1506&0.2177&0.0993&0.1209&0.0282&0.0471&0.0177&0.0238\\
\rowcolor{black!6}
\quad \textbf{\textit{Improv.}} & \textbf{+13.5\%} & \textbf{+13.8\%} & \textbf{+14.2\%} & \textbf{+14.1\%} & \textbf{+8.9\%} & \textbf{+7.4\%} & \textbf{+6.8\%} & \textbf{+6.2\%} & \textbf{+25.3\%} & \textbf{+25.6\%} & \textbf{+28.3\%} & \textbf{+27.3\%}\\
LESSRec+ (LETTER) & 0.0358 & 0.0577 & 0.0231 & 0.0301 & 0.1616 & 0.2270 & 0.1109 & 0.1320 & 0.0324 & 0.0511 & 0.0213 & 0.0273\\
\rowcolor{black!6}
\quad \textbf{\textit{Improv.}} & \textbf{+7.2\%} & \textbf{+12.5\%} & \textbf{+6.0\%} & \textbf{+9.1\%} & \textbf{+16.8\%} & \textbf{+12.0\%} & \textbf{+19.2\%} & \textbf{+16.0\%} & \textbf{+44.0\%} & \textbf{+36.3\%} & \textbf{+54.3\%} & \textbf{+46.0\%}\\
\bottomrule
\end{tabular}}
\end{table*}

Table~\ref{tab:main} compares LESSRec and LESSRec+ with conventional sequential and generative recommenders on Games, Software, and Beauty. We summarize the main observations as follows.
\begin{itemize}[leftmargin=*,nosep]
\item \textbf{Comparison with existing recommenders.} Generative methods based on semantic IDs show an advantage over conventional sequential recommenders. Their semantic-token representations allow related items to share information while expressing item semantics at multiple granularities, which can benefit next-item generation. Building on these representations, the best LESSRec or LESSRec+ configuration achieves the highest score in every dataset--metric column. This supports explicitly modeling the different predictive roles of historical semantic levels to better exploit semantic IDs.
\item \textbf{Consistent gains across tokenizers.} Both models outperform the corresponding baseline under all three tokenizers across all datasets and metrics. Relative gains range from 6.2\% to 39.2\% for LESSRec and from 6.0\% to 54.3\% for LESSRec+. The improvements are not confined to a particular quantizer. These results suggest that models that encode historical sequences separately by semantic level and use same-level history for prediction are better suited to the hierarchical structure of semantic IDs obtained through residual quantization.
\item \textbf{Benefits of level-specific encoders.} LESSRec implements level isolation through attention masking, a simple approach in which all semantic levels still share the encoder layers. LESSRec+ instead assigns separate encoder-layer parameters to each semantic-level history while keeping the total parameter count unchanged. It improves over LESSRec in 31 of the 36 dataset--tokenizer--metric comparisons. This overall trend suggests that fully separating encoder-layer parameters helps each encoder capture its assigned level's semantics while limiting cross-level interference. This separation appears more effective than masking attention interactions alone while sharing attention projections and feed-forward network parameters across levels. The computational benefits of LESSRec+ are examined in Section~\ref{sec:resource_efficiency}.
\end{itemize}

\subsection{Ablation Study (RQ2)}
\label{sec:ablations}
Table~\ref{tab:ablation} examines the three main components, LIS, LAC, and ILS, on Games and Beauty. The \textit{+ILS} variant adds only the item-level semantic sequence to TIGER. The remaining variants progressively add LIS, LAC, and ILS to obtain LESSRec.

Adding ILS alone improves TIGER on both datasets, but its results remain below all variants that include LIS. Thus, supplying item-level context alone does not account for the full model's gains. Adding LIS improves performance even with unrestricted decoder access, suggesting that separating semantic levels during encoding helps retain useful same-level contextual information. Adding LAC yields further gains, indicating that explicitly selecting same-level history benefits prediction beyond encoder-side isolation alone.

Finally, adding ILS to LIS+LAC achieves the best performance among the compared variants on both datasets. Together with the ILS-only results, this supports its complementary role: the interaction-ordered item-level semantic sequence supplies overall item semantics alongside the same-level historical information learned and accessed through LIS and LAC. Additional design-choice ablations are presented in Appendix~\ref{app:additional_experiments}.

\begin{table}[t]
\caption{Component ablations on Games and Beauty.}
\label{tab:ablation}
\centering
\small
\setlength{\tabcolsep}{2pt}
\renewcommand{\arraystretch}{1.10}
\begin{tabular*}{\columnwidth}{@{\extracolsep{\fill}}>{\raggedright\arraybackslash}p{0.48\columnwidth}rrrr@{}}
\toprule
\multirow{2}{*}{Variant} & \multicolumn{2}{c}{Games} & \multicolumn{2}{c}{Beauty}\\
\cmidrule(lr){2-3}\cmidrule(lr){4-5}
& R@10 & N@10 & R@10 & N@10\\
\midrule
\textbf{TIGER (baseline)} & 0.0371 & 0.0208 & 0.0331 & 0.0167\\
\midrule
\quad + ILS & 0.0408 & 0.0231 & 0.0344 & 0.0169\\
\quad + LIS & 0.0418 & 0.0234 & 0.0374 & 0.0191\\
\quad + LIS + LAC & 0.0422 & 0.0239 & 0.0395 & 0.0200\\
\quad + LIS + LAC + ILS (\textbf{LESSRec}) & \textbf{0.0457} & \textbf{0.0257} & \textbf{0.0440} & \textbf{0.0227}\\
\bottomrule
\end{tabular*}
\end{table}

\subsection{Effect of Different Settings (RQ3)}
\label{sec:settings_analysis}
\subsubsection{Fusion Coefficient $\alpha$}
\label{sec:alpha_sensitivity}
The decoder combines the level-aligned token history and the item-level semantic sequence using the coefficient $\alpha$ in Eq.~\ref{eq:dual_mix}. We vary $\alpha$ from 0 (only the level-aligned token history) to 1 (only the item-level semantic sequence), keeping the other settings fixed. Figure~\ref{fig:alpha_sensitivity} shows the results on Games and Beauty. Both datasets achieve their best results with intermediate coefficients, supporting the complementary roles of same-level historical information and overall item semantics.

\begin{figure}[t]
\centering
\includegraphics[width=\columnwidth]{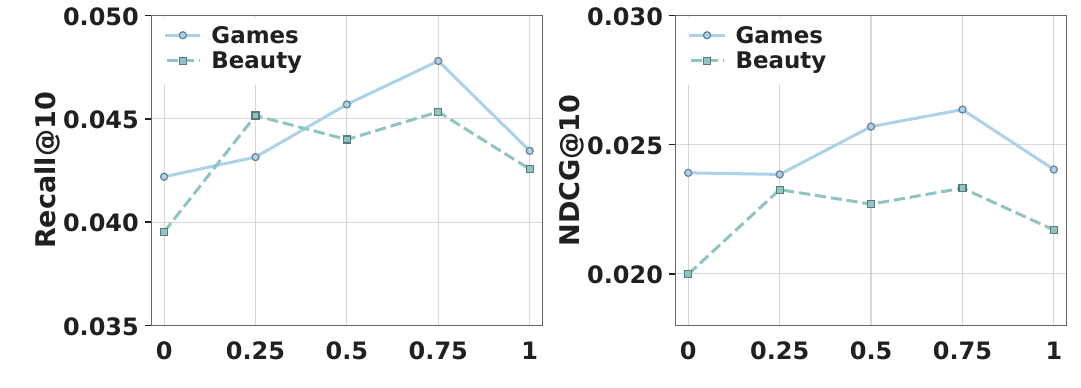}
\caption{Sensitivity to the cross-attention fusion coefficient $\alpha$ on Games and Beauty.}
\Description{Two line-plot panels showing Recall at 10 and NDCG at 10 as functions of the cross-attention fusion coefficient alpha on Games and Beauty. Both datasets achieve their best results with intermediate coefficients.}
\label{fig:alpha_sensitivity}
\end{figure}

\subsubsection{Number of Semantic Tokens per Item}
\label{sec:sid_count_analysis}
Figure~\ref{fig:sid_count_ablation} shows the results on Games with four, five, or six semantic tokens per item. For LESSRec+, the encoder depth allocations from coarse to fine semantic levels are $[2,1,1,1,1]$ for five tokens and $[1,1,1,1,1,1]$ for six tokens. Both LESSRec and LESSRec+ consistently outperform TIGER across all three settings. These results suggest that modeling historical semantic levels benefits hierarchical semantic IDs with different token counts, rather than depending on a particular count.

\begin{figure}[!t]
\centering
\includegraphics[width=\columnwidth]{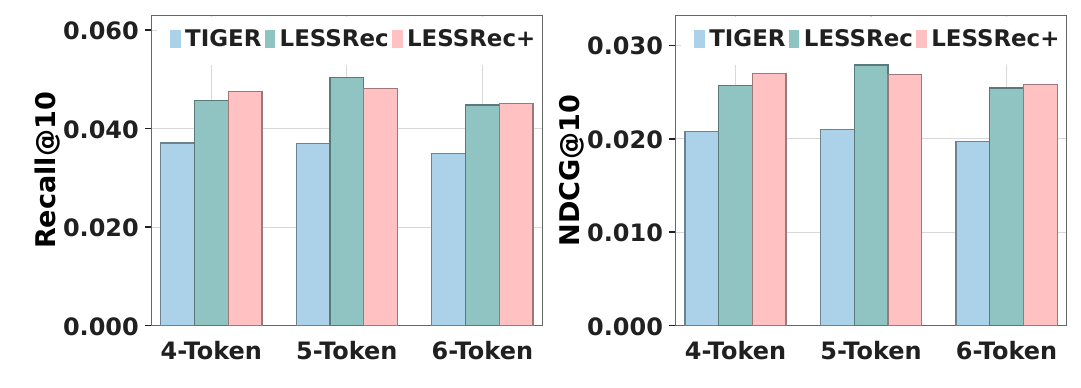}
\caption{Effect of semantic-token count on Games.}
\Description{Two grouped bar charts compare TIGER, LESSRec, and LESSRec+ using four, five, and six semantic tokens per item. The left panel shows Recall at 10, and the right panel shows NDCG at 10. Each group contains three touching bars ordered from left to right as TIGER, LESSRec, and LESSRec+.}
\label{fig:sid_count_ablation}
\end{figure}

\subsubsection{Encoder Depth Allocation}
\label{sec:encoder_depth_analysis}
Table~\ref{tab:encoder_depth_allocation} shows the results of five encoder depth allocations on Games, each retaining six layers in total. No allocation performs best across all metrics: $[3,1,1,1]$ achieves the highest R@5 and N@5, while $[1,1,2,2]$ achieves the highest R@10; the two share the highest N@10 at the reported precision. The default allocation, $[2,2,1,1]$, remains competitive but is not uniformly optimal. These results do not indicate a consistent advantage from assigning more layers exclusively to earlier or later semantic levels. In contrast, $[1,2,2,1]$ performs worst across all metrics, suggesting that allocating additional depth to the middle levels is less effective in this setting. Nevertheless, all five allocations outperform TIGER, indicating that LESSRec+'s gains persist across the tested depth distributions. Choosing how to distribute encoder capacity across semantic levels remains a useful direction for further optimization.

\begin{table}[t]
\caption{LESSRec+ encoder layer allocations on Games.}
\label{tab:encoder_depth_allocation}
\centering
\small
\setlength{\tabcolsep}{2pt}
\renewcommand{\arraystretch}{1.10}
\begin{tabular*}{\columnwidth}{@{\extracolsep{\fill}}lrrrr@{}}
\toprule
Allocation & R@5 & R@10 & N@5 & N@10\\
\midrule
{[2,2,1,1]} & 0.0335 & 0.0475 & 0.0225 & 0.0270\\
{[3,1,1,1]} & 0.0340 & 0.0469 & 0.0230 & 0.0271\\
{[1,2,2,1]} & 0.0312 & 0.0459 & 0.0210 & 0.0257\\
{[1,1,2,2]} & 0.0332 & 0.0481 & 0.0222 & 0.0271\\
{[1,1,1,3]} & 0.0330 & 0.0478 & 0.0217 & 0.0266\\
\bottomrule
\end{tabular*}
\end{table}

\subsection{Complexity and Resource Efficiency (RQ4)}
\label{sec:resource_efficiency}

Table~\ref{tab:attention_complexity} compares the computational complexity of attention operations. The expressions exclude linear projections, feed-forward layers, and repeated beam-search steps. Let $D_e$ and $D_d$ denote the total encoder and decoder depths, $Q$ the number of decoder query positions, $L$ the number of semantic levels, $T$ the number of historical items, and $d$ the hidden dimension, with EOS and padding omitted from the sequence-length analysis.

\begin{table}[!t]
\caption{Complexity of attention-matrix operations.}
\label{tab:attention_complexity}
\centering
\small
\setlength{\tabcolsep}{2pt}
\renewcommand{\arraystretch}{1.10}
\begin{tabular*}{\columnwidth}{@{\extracolsep{\fill}}lcc@{}}
\toprule
Model & \shortstack{Encoder\\self-attention} & \shortstack{Decoder\\cross-attention}\\
\midrule
TIGER & $\mathcal{O}(D_eL^2T^2d)$ & $\mathcal{O}(D_dQLTd)$\\
LESSRec & $\mathcal{O}(D_eL^2T^2d)$ & $\mathcal{O}(D_dQLTd)+\mathcal{O}(D_dQTd)$\\
LESSRec+ & $\mathcal{O}(D_eT^2d)$ & $\mathcal{O}(D_dQLTd)+\mathcal{O}(D_dQTd)$\\
\bottomrule
\end{tabular*}
\end{table}

TIGER and LESSRec compute self-attention over $LT$ historical positions in each encoder layer, giving $\mathcal{O}(D_eL^2T^2d)$. LESSRec+ processes sequences of length $T$, with the level-specific encoder depths summing to $D_e$, giving $\mathcal{O}(D_eT^2d)$. This yields a factor of $L^2$ reduction in these operations.

Decoder cross-attention over the encoded token history costs $\mathcal{O}(D_dQLTd)$. Both LESSRec variants retain this branch and add an item-level branch costing $\mathcal{O}(D_dQTd)$. The decoder cost is therefore the sum of these two terms, with the additional branch scaling linearly with history length. The main reduction in attention computation comes from the level-specific encoders in LESSRec+.

Table~\ref{tab:history_length_performance_efficiency} compares model parameters, time to convergence, average training time per epoch, and peak training memory for TIGER, LESSRec, and LESSRec+. Both proposed models add only 0.4\% more parameters, primarily from the item-representation MLP and a small number of item-level positional parameters. Time to convergence decreases by 34.0\% for LESSRec and 65.2\% for LESSRec+, consistent with the models learning useful predictive patterns more readily from separately encoded same-level histories.

LESSRec takes 5.8\% longer per epoch because attention masking retains the original attention computation, while the item-level semantic sequence adds representation construction and a second cross-attention branch. In contrast, LESSRec+ reduces per-epoch time by 33.7\% and peak training memory by 39.2\%. Its level-specific encoders process shorter histories independently, avoiding cross-level attention computations that LESSRec still performs before masking.

\begin{table}[!t]
\caption{Model parameters and training costs. Colored arrows indicate percentage increases (red) or decreases (green) relative to TIGER.}
\label{tab:history_length_performance_efficiency}
\centering
\small
\setlength{\tabcolsep}{2pt}
\renewcommand{\arraystretch}{1.10}
\begin{tabular*}{\columnwidth}{@{\extracolsep{\fill}}lcccc@{}}
\toprule
Model & \shortstack{Params\\(M)} & \shortstack{Training time\\to convergence (h)} & \shortstack{Training time\\per epoch (s)} & \shortstack{Peak\\Memory (GiB)}\\
\midrule
TIGER & 60.86 & 2.27 & 99.53 & 8.79\\
\midrule
LESSRec & \shortstack{61.12\\\textcolor{red}{\(\uparrow\)\,0.4\%}} & \shortstack{1.49\\\textcolor{green!50!black}{\(\downarrow\)\,34.0\%}} & \shortstack{105.33\\\textcolor{red}{\(\uparrow\)\,5.8\%}} & \shortstack{9.46\\\textcolor{red}{\(\uparrow\)\,7.5\%}}\\
\midrule
LESSRec+ & \shortstack{61.12\\\textcolor{red}{\(\uparrow\)\,0.4\%}} & \shortstack{0.79\\\textcolor{green!50!black}{\(\downarrow\)\,65.2\%}} & \shortstack{65.95\\\textcolor{green!50!black}{\(\downarrow\)\,33.7\%}} & \shortstack{5.35\\\textcolor{green!50!black}{\(\downarrow\)\,39.2\%}}\\
\bottomrule
\end{tabular*}
\end{table}

\subsection{Effect of History Length (RQ5)}
\label{sec:history_length_analysis}
We group test cases on Games into three history-length buckets: fewer than 10, 10--14, and 15--20 interactions. Figure~\ref{fig:history_length_performance_efficiency} reports Recall@5 over seed 42 and encoder CUDA time per sample. Bars use the left axes for absolute values, while lines use the right axes for percentage changes relative to TIGER within each bucket. 

Both proposed models outperform TIGER in all three buckets, with larger relative gains for longer histories. LESSRec's gains increase from 23.4\% in the shortest-history bucket to 46.7\% in the longest, while LESSRec+'s gains increase from 29.8\% to 53.3\%. These results suggest that same-level historical modeling remains effective across different history lengths and offers more pronounced benefits for users with longer observed histories. A possible explanation is that level isolation confines each level-aligned prediction to same-level history, only $1/L$ of all historical tokens, which limits dilution from irrelevant tokens as histories grow. The encoding-time results distinguish the two implementations. LESSRec remains close to TIGER, with changes ranging from a 1.6\% reduction to a 5.1\% increase. LESSRec+ incurs a 9.2\% overhead in the shortest-history bucket but reduces encoding time by 13.4\% and 29.5\% in the middle and longest buckets, respectively. This pattern is consistent with the computational advantage of processing shorter, level-specific sequences becoming more consequential as history length grows.

\begin{figure}[!t]
\centering
\includegraphics[width=\columnwidth,trim=1.526698bp 43bp 3.963257bp 20bp]{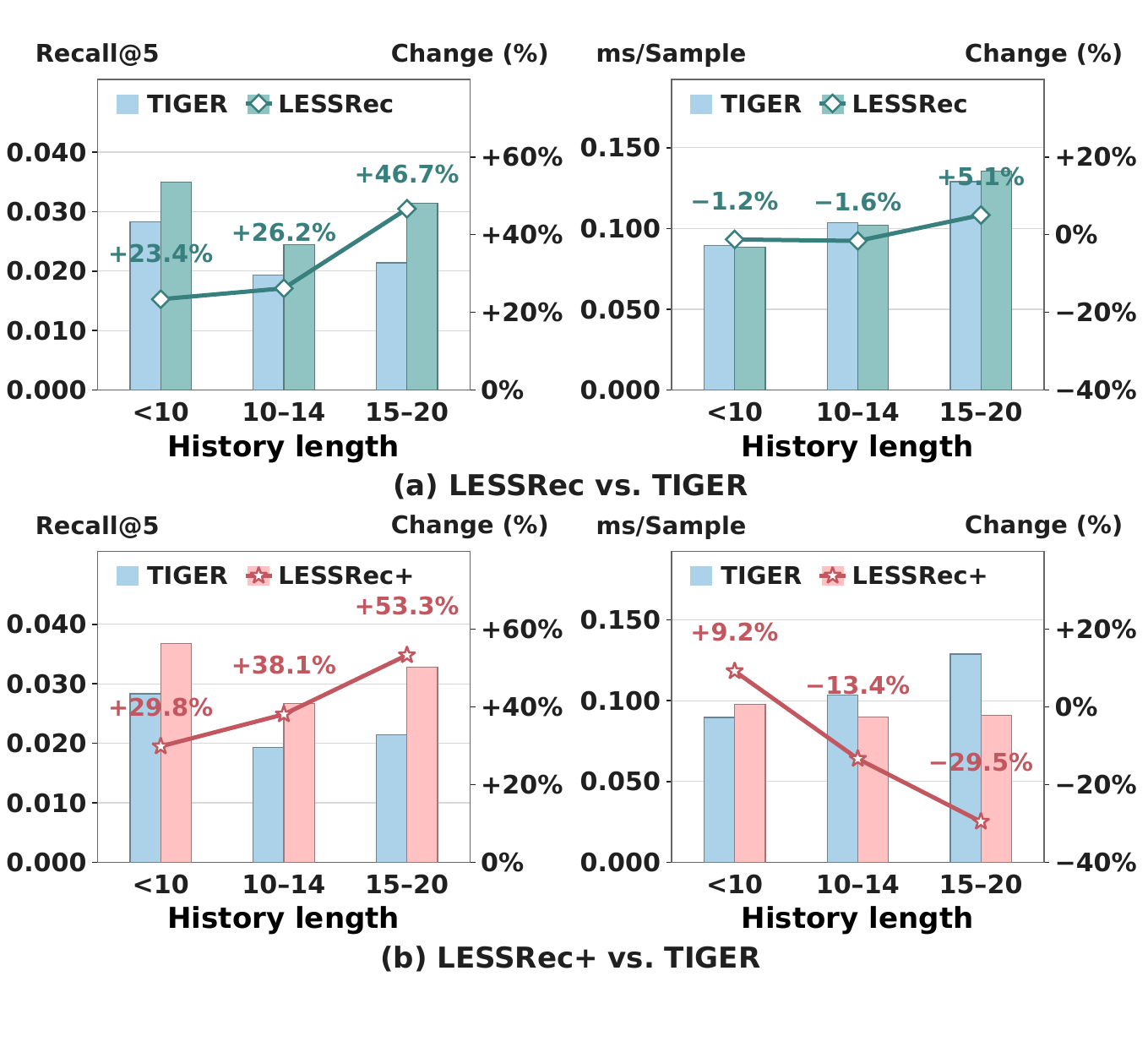}
\caption{History-length analysis: (a) LESSRec versus TIGER; (b) LESSRec+ versus TIGER. Each row shows Recall@5 (left) and encoder CUDA time per sample (right).}
\label{fig:history_length_performance_efficiency}
\end{figure}

\section{Related Work}
\label{sec:related_work}
\subsection{Semantic Item Tokenization}
{\emergencystretch=1em
Semantic item tokenization determines the information encoded in discrete identifiers. VQ-Rec~\cite{hou2023learning} discretizes item text embeddings using product quantization (PQ)~\cite{jegou2011product}, while TIGER~\cite{rajput2023recommender} obtains hierarchical semantic IDs through RQ-VAE~\cite{lee2022autoregressive}. Related approaches explore learnable document identifiers~\cite{sun2023learning}, collaborative-aware item tokenization~\cite{qu2025tokenrecb}, tree-structured codebooks~\cite{zheng2026universal}, textual identifiers~\cite{tan2024idgenreca}, and action-context tokenization~\cite{hou2025actionpiece}. LESSRec keeps semantic IDs fixed and studies their use in historical sequence modeling. We evaluate TIGER, LESSRec, and LESSRec+ with RQ-VAE, RQ-Kmeans, and LETTER tokenization.\par}

\subsection{Semantic ID-Based Generative Recommendation}
{\emergencystretch=1em
Semantic ID-based generative recommenders predict the next item's identifier from user history. TIGER adopts autoregressive generation of residual-quantized IDs~\cite{rajput2023recommender}. ETEGRec couples tokenizer and recommender learning through sequence--item and preference--semantic alignment with alternating optimization~\cite{liu2025generative}. EAGER-LLM integrates external behavioral and semantic information into a decoder-only language model through dual-source item indices and multiscale alignment tasks~\cite{hong2025eagerllm}. RPG generates long, unordered semantic IDs in parallel with graph-constrained decoding~\cite{hou2025generating}, while CARE combines progressive history encoding with query-anchored reasoning to mitigate bias amplification~\cite{lin2026bringing}. These methods improve joint learning, information integration, or generation strategies, complementing our focus on historical semantic-level modeling.\par}

\subsection{Structure-Aware Sequence Modeling}
{\emergencystretch=1em
Structure-aware methods incorporate the hierarchy of semantic IDs and item boundaries into sequence modeling. TrieRec encodes prefix-tree topology through absolute and relative positional encodings~\cite{xu2026trieaware}. DiscRec combines intra-item attention with embedding-level fusion of semantic and collaborative representations~\cite{liu2025discreca}. Hi-SAM compresses items into anchor tokens and separates inter-item from intra-item positional information~\cite{pan2026hisam}. BARGE restores item context during encoding and mitigates semantic drift through hierarchical path reranking and dual-path decoding~\cite{zeng2026bridging}. LESSRec encodes same-level histories across items and aligns token-history access with the target semantic level, while retaining item-level context through a separate cross-attention branch. We compare LESSRec and LESSRec+ with DiscRec~\cite{liu2025discreca}, a representative structure-aware baseline, in Section~\ref{sec:experiments}.\par}

\section{Conclusion}
In this paper, we show that historical semantic levels play different roles in target-token prediction, with same-level history generally being more informative. LESSRec separately encodes historical token sequences at each semantic level and aligns decoder access with the target level. An item-level semantic sequence combines contextualized token states across levels within each item, providing complementary item semantics through dual cross-attention. LESSRec+ reduces encoder computation through level-specific encoder groups while retaining the same parameter count as LESSRec. Experiments on three public datasets with three tokenizers demonstrate both models' recommendation effectiveness, while component ablations and efficiency analyses support the proposed design and LESSRec+'s computational benefits.

\bibliographystyle{ACM-Reference-Format}
\bibliography{references}
\appendix
\section{Why Same-Level History Helps}
\label{app:level_theory}
\begingroup
\emergencystretch=2em
We seek to explain why same-level history can be more useful for target-token prediction from two perspectives: predictive information and learning signal. We first introduce assumptions on the residual structure and temporal dependence, then examine their implications for prediction and learning.

\paragraph{Assumptions.}
First, we assume that successive residual components are sufficiently separated for quantization to recover them. Write the tokenizer latent as $\mathbf{u}_i=\sum_{r=1}^{L}\mathbf{e}_{r,Z_{i,r}}+\boldsymbol{\xi}_i$, where $Z_{i,r}$ is a discrete component. With minimum codeword separation $\Delta_r=\min_{a\ne b}\|\mathbf{e}_{r,a}-\mathbf{e}_{r,b}\|_2>0$, assume
\begin{equation}
\left\|\sum_{j>r}\mathbf{e}_{j,Z_{i,j}}+\boldsymbol{\xi}_i\right\|_2
<\Delta_r/2\qquad\text{for all }i,r.
\label{eq:theory_separation}
\end{equation}
This condition ensures $s_{i,r}=Z_{i,r}$. Inductively, the residual is the correct codeword plus a tail $\mathbf{t}$ with $\|\mathbf{t}\|_2<\Delta_r/2$. Every alternative is farther away, at distance at least $\Delta_r-\|\mathbf{t}\|_2>\|\mathbf{t}\|_2$. Scalar codebooks $\{-3^{L-r},3^{L-r}\}$ with zero final residual provide one example satisfying this condition.

Second, we assume that the component trajectories $Z_{1:T+1,r}$ are independent across levels, with each following a first-order Markov chain. This describes temporal dependence within each residual level; RQ-VAE training alone does not guarantee it. Write $\mathcal{H}=(H_1,\ldots,H_L)$, $Y_k=s_{T+1,k}$, and $C_k=Y_{<k}$. All logarithms are natural.

\paragraph{Predictive information.}
Under these assumptions, we can compare the predictive value of different histories through their lowest achievable token log loss:
\begin{equation*}
\mathcal{R}^{*}(A)=\inf_q\mathbb{E}[-\log q(Y_k\mid C_k,A)].
\end{equation*}
Here, $q$ ranges over all conditional predictors. The following proposition relates the benefit of same-level history to its temporal predictive information.

\begin{appcproposition}[Predictive information]
\label{thm:level_information}
Under the stated assumptions, for $r\ne k$,
\begin{equation}
\begin{aligned}
\mathcal{R}^{*}(\mathcal{H})&=\mathcal{R}^{*}(H_k)
=\mathrm{H}(Y_k\mid s_{T,k}),\\
\mathcal{R}^{*}(H_r)-\mathcal{R}^{*}(H_k)
&=\mathrm{I}(Y_k;s_{T,k})\geq0.
\end{aligned}
\label{eq:theory_information_gap}
\end{equation}
The gap is positive when the previous same-level token carries predictive information about the target.
\end{appcproposition}
\begin{proof}
Expected log loss equals conditional entropy plus a nonnegative expected KL divergence, so its infimum is conditional entropy. Independence of level trajectories and the Markov property give
\begin{equation*}
\begin{aligned}
P(Y_k\mid C_k,\mathcal{H})&=P(Y_k\mid s_{T,k}),\\
P(Y_k\mid C_k,H_r)&=P(Y_k),\qquad r\ne k.
\end{aligned}
\end{equation*}
Substituting these distributions into the conditional entropies gives the first equality and the mutual-information gap.
\end{proof}

Same-level history therefore retains the predictive information of the full history under these assumptions. Any advantage over full-history attention must arise from how a restricted model uses that information.

\paragraph{Learning signal.}
We next examine this distinction with a single-head model. Use the binary construction above with $L>1$ and $T=1$. Let historical signs $X_r$ be independent and uniform, and set $Y=X_kE$ with independent noise $\Pr(E=-1)=p\in(0,1/2)$. Other levels have uniform next states independent of the history and $E$, so $C_k$ is independent of $(X_1,\ldots,X_L,Y)$. A one-layer, level-isolated encoder with fixed sign embeddings and identity value projection returns $X_r$. A single decoder head uses fixed finite scores depending only on level, with $a_r=\operatorname{softmax}(\ell)_r>0$. Its two contexts are
\begin{equation}
z_{\mathrm{same}}=X_k,\qquad
z_{\mathrm{all}}=mX_k+N,\quad
m=a_k,\quad N=\sum_{r\ne k}a_rX_r.
\label{eq:theory_binary_context}
\end{equation}
Here, $N$ collects the other levels' contributions. The classifier is $\Pr_w(Y=1\mid z)=\sigma(wz)$, where $|w|\leq W$, $W>0$, and $\sigma$ is the sigmoid. We omit layer normalization, feed-forward sublayers, EOS, and item memory to isolate the effect of mixing levels. The next proposition compares both prediction loss and the gradient used to learn $w$.

\begin{appcproposition}[Prediction loss and initial learning signal]
\label{prop:learning_advantage}
Let $\mathcal{L}(w;z)=\mathbb{E}\log(1+e^{-Ywz})$. Then
\begin{equation}
\inf_{|w|\leq W}\mathcal{L}(w;z_{\mathrm{all}})
\geq\inf_{|w|\leq W}\mathcal{L}(w;z_{\mathrm{same}}).
\label{eq:theory_binary_risk_bound}
\end{equation}
At $w=0$, define the per-example gradient $G=\partial_w\log(1+e^{-Ywz})$ and $\tau=1-2p$. For the average $\overline G_n$ of $n$ independent examples, let $\mathrm{SNR}_n=(\mathbb{E}\overline G_n)^2/\operatorname{Var}(\overline G_n)$. Then
\begin{equation}
\begin{aligned}
\mathbb{E}G_{\mathrm{all}}&=-m\tau/2,
&\mathbb{E}G_{\mathrm{same}}&=-\tau/2,\\
\mathrm{SNR}^{\mathrm{all}}_n
&=\frac{nm^2\tau^2}{m^2(1-\tau^2)+\sum_{r\ne k}a_r^2}
& &\\[-0.5em]
&<\frac{n\tau^2}{1-\tau^2}
=\mathrm{SNR}^{\mathrm{same}}_n.
\end{aligned}
\label{eq:theory_gradient_snr}
\end{equation}
\end{appcproposition}
\begin{proof}
For the loss comparison, $N$ is independent of $(X_k,Y)$ and has mean zero. Conditional Jensen's inequality therefore gives $\mathcal{L}(w;mX_k+N)\geq\mathcal{L}(mw;X_k)$. Taking infima over $[-mW,mW]\subseteq[-W,W]$ proves the bound. For the learning comparison, $G=-Yz/2$ at initialization. Independence yields
\begin{equation*}
\operatorname{Var}(G_{\mathrm{all}})
=\tfrac14\left[m^2(1-\tau^2)+\sum_{r\ne k}a_r^2\right].
\end{equation*}
The aligned case has $m=1$ and $N=0$. Averaging divides the variance by $n$; positive mismatched weights make the SNR inequality strict.
\end{proof}

The second proposition explains how mixing levels can weaken the learning signal even when all predictive information remains available in the history. With equal scores, the expected gradient magnitude falls by a factor of $L$, and $\mathrm{SNR}^{\mathrm{all}}_n=n\tau^2/(L-\tau^2)$. The analysis thus supports same-level access under the stated assumptions. Its learning conclusion concerns initialization, not global convergence; learned projections and content-dependent scores may change the comparison.
\par\endgroup
\label{app:level_theory_end}

\section{Residual Quantization for Item Tokenization}
\label{app:tokenization}
We describe the offline tokenizers used to construct the fixed semantic IDs. Let $\mathbf{a}_i$ denote the content embedding of item $v_i$, $L$ the number of semantic tokens, and $K$ the number of token values at each level. Tokenizer parameters are trained separately from the recommendation model.

\subsection{RQ-VAE}
\label{sec:rqvae}
Following residual quantization for image generation~\cite{lee2022autoregressive} and its use in TIGER~\cite{rajput2023recommender}, RQ-VAE encodes the content embedding into a continuous latent vector $\mathbf{u}_i=\mathcal{E}_{\mathrm{q}}(\mathbf{a}_i)$. Each level $r$ has a codebook $\mathcal{C}_r=\{\mathbf{e}_{r,k}\}_{k=0}^{K-1}$. Starting with $\boldsymbol{\rho}_{i,0}=\mathbf{u}_i$, quantization selects the nearest codeword and subtracts it from the residual:
\begin{equation}
\begin{aligned}
s_{i,r}&=\arg\min_{0\leq k<K}
\|\boldsymbol{\rho}_{i,r-1}-\mathbf{e}_{r,k}\|_2^2,\\
\boldsymbol{\rho}_{i,r}&=\boldsymbol{\rho}_{i,r-1}-\mathbf{e}_{r,s_{i,r}},
\qquad r=1,\ldots,L.
\end{aligned}
\label{eq:app_rq_assign}
\end{equation}
The selected indices form $\mathbf{s}_i=(s_{i,1},\ldots,s_{i,L})$. Earlier codewords approximate the latent representation, while later codewords describe the remaining residual, producing a coarse-to-fine semantic-level structure. The quantized representation and reconstructed content embedding are
\begin{equation}
\widehat{\mathbf{u}}_i=\sum_{r=1}^{L}\mathbf{e}_{r,s_{i,r}},
\qquad
\widehat{\mathbf{a}}_i=\mathcal{D}_{\mathrm{q}}(\widehat{\mathbf{u}}_i).
\label{eq:app_rq_reconstruct}
\end{equation}
The encoder, decoder, and codebooks are jointly trained using reconstruction, codebook, and commitment losses~\cite{vandenoord2017neurala}:
\begin{equation}
\begin{aligned}
\ell_{i,r}^{\mathrm{vq}}={}&
\operatorname{MSE}(\operatorname{sg}[\boldsymbol{\rho}_{i,r-1}],
\mathbf{e}_{r,s_{i,r}})\\
&+\beta\operatorname{MSE}(\boldsymbol{\rho}_{i,r-1},
\operatorname{sg}[\mathbf{e}_{r,s_{i,r}}]),\\
\mathcal{L}_{\mathrm{q}}={}&\frac{1}{N}\sum_{i=1}^{N}\left[
\operatorname{MSE}(\mathbf{a}_i,\widehat{\mathbf{a}}_i)
+\frac{\lambda}{L}\sum_{r=1}^{L}\ell_{i,r}^{\mathrm{vq}}\right].
\end{aligned}
\label{eq:app_rq_loss}
\end{equation}
Here, $N$ is the number of items, $\operatorname{MSE}$ averages squared errors over vector coordinates, and $\operatorname{sg}$ denotes stop-gradient. The codebook term fits selected vectors to the residuals, while the commitment term encourages the encoded residuals to remain close to their selected vectors. The implementation uses $\beta=0.25$, $\lambda=1$, a straight-through estimator for quantization, and K-means initialization of the codebooks from an initial batch of encoded residuals. Codebooks are subsequently updated by gradient descent; K-means initialization alone does not make this tokenizer a residual K-means tokenizer.

\subsection{RQ-Kmeans}
\label{app:rqkmeans}
RQ-Kmeans constructs hierarchical semantic IDs by clustering residual representations at successive levels~\cite{luo2025qarm,hu2026once}. For the residuals entering level $r$, standard residual K-means learns a codebook through
\begin{equation}
\min_{\mathcal{C}_r}\sum_{i=1}^{N}\min_{\mathbf{e}\in\mathcal{C}_r}
\|\boldsymbol{\rho}_{i,r-1}-\mathbf{e}\|_2^2.
\label{eq:app_rqkmeans_objective}
\end{equation}
Each item is assigned to its nearest centroid, which is subtracted before quantizing the next residual. The resulting codeword indices form the semantic IDs, and the sum of the selected centroids approximates the input representation. Earlier levels capture larger-scale variation, while later levels refine the remaining residual.

Our implementation combines this residual formulation with a learned autoencoder and K-means-guided conditional subcodebooks. Let $\mathcal{C}_{r,g}=\{\mathbf{e}_{r,g,k}\}_{k=0}^{K_{r,g}-1}$ denote subcodebook $g$ at level $r$, containing $K_{r,g}$ vectors. Earlier semantic tokens determine the subcodebook through a routing function $\Gamma_r$:
\begin{equation}
\begin{aligned}
g_{i,r}&=\Gamma_r(s_{i,1},\ldots,s_{i,r-1}),\\
k_{i,r}&=\arg\min_{0\leq k<K_{r,g_{i,r}}}
\|\boldsymbol{\rho}_{i,r-1}-\mathbf{e}_{r,g_{i,r},k}\|_2^2,\\
\boldsymbol{\rho}_{i,r}&=\boldsymbol{\rho}_{i,r-1}
-\mathbf{e}_{r,g_{i,r},k_{i,r}}.
\end{aligned}
\label{eq:app_conditional_quantization}
\end{equation}
The initial residual is $\boldsymbol{\rho}_{i,0}=\mathcal{E}_{\mathrm{q}}(\mathbf{a}_i)$, and the first-level subcodebook is shared by all items. The selected subcodebook and local codeword index jointly identify the semantic token at each level. As in RQ-VAE, the selected vectors are summed and decoded to reconstruct the content embedding.

\section{Additional Experimental Analysis}
\label{app:additional_experiments}
Figure~\ref{fig:ablation_games_design} further examines LESSRec design variants on Games. \textit{w/o MLP} directly uses the mean-pooled representations as item vectors, while \textit{Cat.+MLP} replaces within-item mean pooling with concatenation before the MLP. Removing the MLP reduces performance in both metrics, supporting its role in refining pooled representations for recommendation. In contrast, concatenation followed by the MLP outperforms the mean-pooling-based LESSRec reference in both metrics. This suggests that retaining the individual semantic-level representations before learned fusion can benefit item-level modeling; mean pooling is therefore not necessarily the best aggregation choice.

We also evaluate \textit{w/o EOS}, which removes EOS from the encoder input. Performance remains close to LESSRec in both metrics. Although the EOS state can aggregate information across semantic levels and is accessible to the decoder, its removal produces no consistent performance gain. These results suggest that a single shared EOS state, despite its access to all semantic levels, does not introduce cross-level interference strong enough to noticeably affect recommendation performance in this setting. Historical semantic tokens cannot attend to EOS under LIS, so their encoded states remain level-isolated even when EOS is retained as the sequence terminator.

We further compare \textit{QKV LIS}, \textit{QKV LAC}, and \textit{QKV LIS+LAC}, which compute attention in separate semantic-level blocks in encoder self-attention, decoder cross-attention, or both, respectively, while retaining shared QKV projections. Splitting encoder-side QKV computation has little effect on performance, whereas splitting decoder-side computation leads to a clear decline. A possible explanation concerns attention to the shared item-level memory: in LESSRec, target positions at different semantic levels access the same item-level semantic sequence through a shared cross-attention module. Reorganizing decoder attention into separate level-specific blocks may affect how these queries access and combine the shared item-level context with their corresponding same-level histories, potentially weakening the complementarity of the two branches.

\begin{figure}[H]
\centering
\includegraphics[width=\columnwidth]{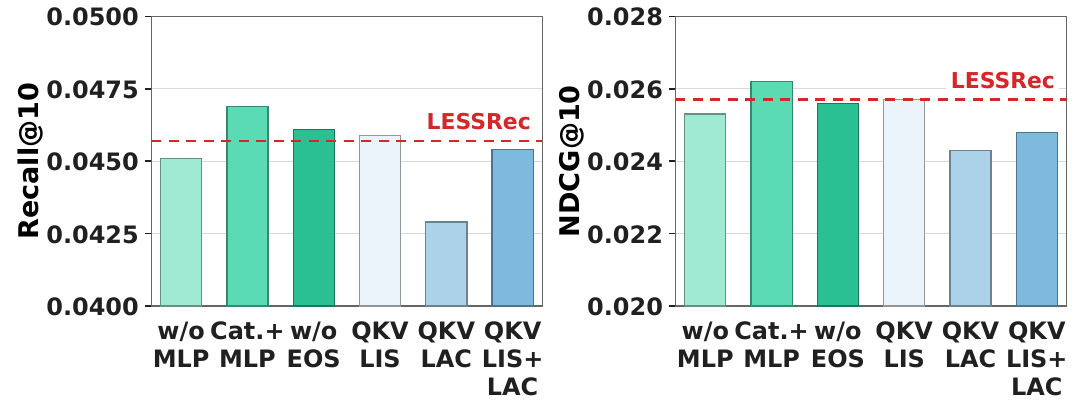}
\caption{Ablation of LESSRec design choices on Games.}
\Description{Two bar-chart panels show Recall at 10 and NDCG at 10 on Games for six LESSRec variants: mean pooling without an MLP, concatenation followed by an MLP, removal of EOS from the encoder input, and blockwise attention computation in encoder self-attention, decoder cross-attention, or both while retaining shared QKV projections. Red dashed lines mark full LESSRec. Concatenation followed by an MLP exceeds the reference in both panels, removing EOS and encoder-side block computation give performance close to the reference, and decoder-side block computation yields the lowest performance.}
\label{fig:ablation_games_design}
\end{figure}

\balance
\end{document}